\documentclass[10pt]{article}

\usepackage{hyperref}
\usepackage{amsmath,amsthm}
\usepackage{fullpage}
\usepackage{amsfonts}
\usepackage{changepage}
\usepackage{bookmark}
\usepackage{graphicx}
\usepackage[absolute]{textpos}
\def\boxit#1{\vbox{\hrule\hbox{\vrule\kern4pt
  \vbox{\kern1pt#1\kern1pt}
\kern2pt\vrule}\hrule}}
\usepackage{tkz-graph}
\newtheorem{definition}{Definition}
\newtheorem{lemma}{Lemma}
\newtheorem{theorem}{Theorem}

\begin{document}

\title{An Improved FPT Algorithm for the Flip Distance Problem%
\thanks{A preliminary version of this paper appeared in the proceedings of 42nd International Symposium on Mathematical Foundations of Computer Science, MFCS 2017.}
\thanks{This work is supported by the National Natural Science Foundation of China under Grants (61672536, 61420106009, 61872450, 61828205,), Hunan Provincial Science and Technology Program (2018WK4001) and the European Research Council (ERC) under the European Union's Horizon 2020 research and the innovation programme (grant Nr: 714704).}}

\author{
  Qilong Feng\thanks{School of Computer Science and Engineering, Central South University, Changsha, Hunan, P.R.China.}
  \and
  Shaohua Li\thanks{Institute of Informatics, University of Warsaw, Poland.}
  \and
  Xiangzhong Meng\thanks{School of Computer Science and Engineering, Central South University, Changsha, Hunan, P.R.China.}
  \and
  Jianxin Wang\thanks{School of Computer Science and Engineering, Central South University, Changsha, Hunan, P.R.China, \texttt{jxwang@mail.csu.edu.cn}}}
  \date{}

\maketitle

\begin{textblock}{20}(0, 12.0)
\includegraphics[width=40px]{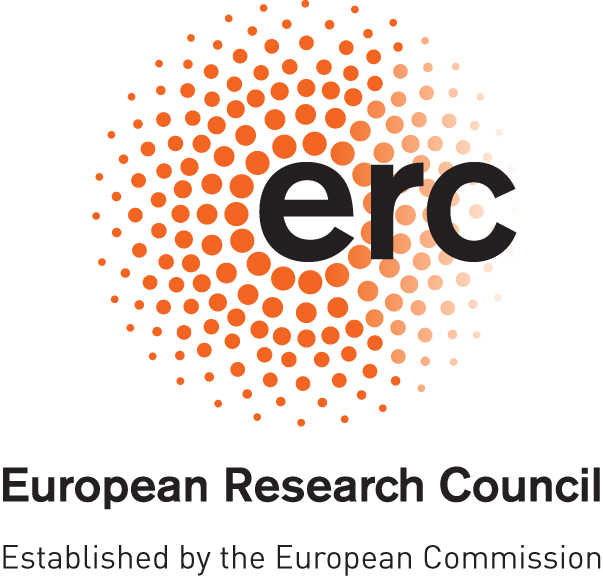}%
\end{textblock}
\begin{textblock}{20}(-0.25, 12.4)
\includegraphics[width=60px]{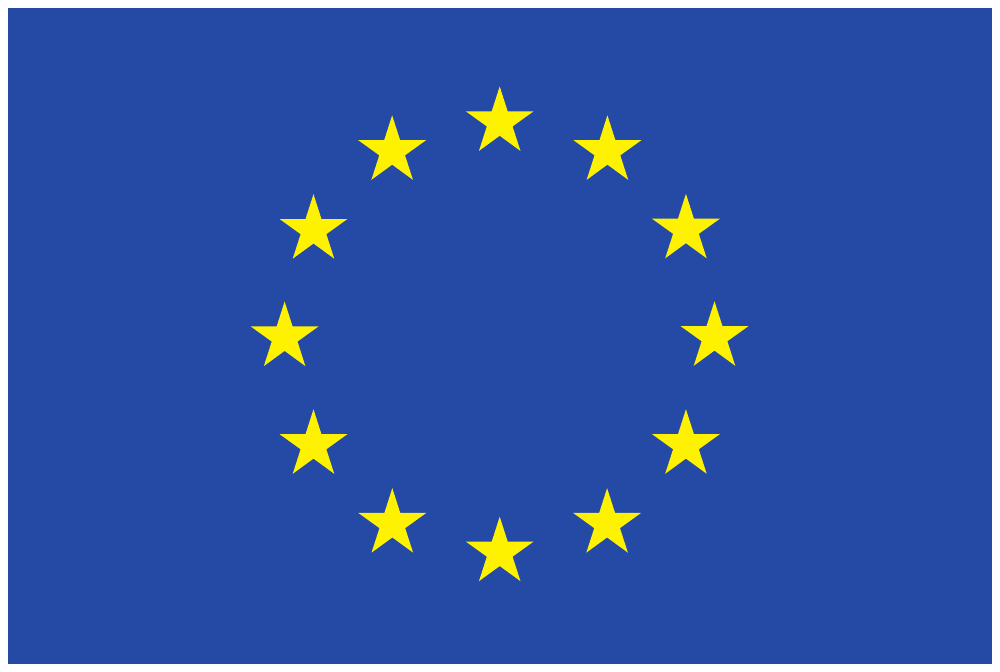}%
\end{textblock}

\begin{abstract}
Given a set $\cal P$ of points in the Euclidean plane and two triangulations of $\cal P$, the flip distance between these two triangulations is the minimum number of flips required to transform one triangulation into the other. \textsc{Parameterized Flip Distance} problem is to decide if the flip distance between two given triangulations is equal to a given integer $k$. The previous best FPT algorithm runs in time $O^{*}(k\cdot c^{k})$ ($c\leq 2\times 14^{11}$), where each step has fourteen possible choices, and the length of the action sequence is bounded by $11k$. By applying the backtracking strategy and analyzing the underlying property of the flip sequence, each step of our algorithm has only five possible choices. Based on an auxiliary graph $G$, we prove that the length of the action sequence for our algorithm is bounded by $2|G|$. As a result, we present an FPT algorithm running in time $O^{*}(k\cdot 32^{k})$.
\end{abstract}

\section{Introduction}\label{Introduction}
Given a set $\cal P$ of $n$ points in the Euclidean plane, a \emph{triangulation} of $\cal P$ is a maximal planar subdivision whose vertex set is $\cal P$ \cite{Mark}. A \emph{flip} operation to one diagonal $e$ of a convex quadrilateral in a triangulation is to remove $e$ and insert the other diagonal into this quadrilateral. Note that if the quadrilateral associated with $e$ is not convex, the flip operation is not allowed.  The \emph{flip distance} between two triangulations is the minimum number of flips required to transform one triangulation into the other.

Triangulations play an important role in computational geometry, which are applied in areas such as computer-aided geometric design and numerical analysis \cite{farin2014curves,Hamann,schumaker1993triangulations}.

Given a point set $\cal P$ in the Euclidean plane, we can construct a graph $G_{T}(\cal P)$ in which every triangulation of $\cal P$ is represented by a vertex, and two vertices are adjacent if their corresponding triangulations can be transformed into each other through one flip operation.  $G_{T}(\cal P)$ is called the \emph{triangulations graph} of $\cal P$.  Properties of the triangulations graph are studied in the literature.  Aichholzer et al. \cite{Aichholzer} showed that the lower bound of the number of vertices of $G_{T}(\cal P)$ is $\Omega (2.33^{n})$.  Lawson and Charles \cite{Lawson} showed that the diameter of $G_{T}(\cal P)$ is $O(n^{2})$. Hurtado et al. \cite{Hurtado} proved that the bound is tight. Since $G_{T}(\cal P)$ is connected \cite{Lawson}, any two triangulations of $\cal P$ can be transformed into each other through a certain number of flips.

\textsc{Flip Distance} problem consists in computing the flip distance between two triangulations of $\cal P$, which was proved to be NP-complete by Lubiw and Pathak \cite{Lubiw}.  Pilz showed that the \textsc{Flip Distance} is APX-hard \cite{Pilz}. Aichholzer et al. \cite{Aichholzer2} proved that \textsc{Parameterized Flip Distance} is NP-complete on triangulations of simple polygons.  However, the complexity of \textsc{Flip Distance} on triangulations of convex polygons has been open for many years, which is equivalent to the problem of computing the rotation distance between two rooted binary trees \cite{Sleator}.

\textsc{Parameterized Flip Distance} problem is: given two triangulations of a set of points in the plane and an integer $k$, deciding if the flip distance between these two triangulations is equal to $k$.  For \textsc{Parameterized Flip Distance} on triangulations of a convex polygon, Lucas \cite{Lucas} gave a kernel of size $2k$ and an $O^{*}(k^{k})$-time algorithm.  Kanj and Xia \cite{Kanj} studied \textsc{Parameterized Flip Distance} on triangulations of a set of points in the plane, and presented an $O^{*}(k\cdot c^{k})$-time algorithm ($c\leq 2\cdot 14^{11}$), which applies to triangulations of general polygonal regions (even with holes or points inside it).

In this paper, we exploit \textsc{Parameterized Flip Distance} further. At first, we give a nondeterministic construction process to illustrate our idea. The nondeterministic construction process contains only two types of actions, which are the \emph{moving action} as well as the \emph{flipping and backing action}.  Given two triangulations and a parameter $k$, we prove that either there exists a sequence of actions of length at most $2k$, following which we can transform one triangulation into the other, or we can conclude that no valid sequence of length $k$ exists.  Thus we get an improved $O^{*}(k\cdot 32^{k})$-time FPT algorithm, which also applies to triangulations of general polygonal regions (even with holes or points inside it).

\section{Preliminaries}
\label{Preliminaries}
In a triangulation $T$, a flip operation $f$ to an edge $e$ that is the diagonal of a convex quadrilateral $\cal Q$ is to delete $e$ and insert the other diagonal $e'$ into $\cal Q$.    We define $e$ as the \emph{underlying edge} of $f$, denoted by $\varepsilon (f)$, and $e'$ as the \emph{resulting edge} of $f$, denoted by $\varphi (f)$. (For consistency and clarity, we continue to use some symbols and definitions from \cite{Kanj}). Note that if $e$ is not a diagonal of any convex quadrilateral in the triangulation, flipping $e$ is not allowed.  Suppose that we perform a flip operation $f$ on a triangulation $T_{1}$ and get a new triangulation $T_{2}$.  We say $f$ transforms $T_{1}$ into $T_{2}$.  $T_{1}$ is called an \emph{underlying triangulation} of $f$, and $T_{2}$ is called a \emph{resulting triangulation} of $f$. Given a set $\cal P$ of $n$ points in the Euclidean plane, let $T_{start}$ and $T_{end}$ be two triangulations of $\cal P$, in which $T_{start}$ is the initial triangulation and $T_{end}$ is the objective triangulation. Let $F = \langle f_{1}, f_{2},...,f_{r}\rangle$ be a sequence of flips, and $\langle T_{0}, T_{1},...,T_{r}\rangle$ be a sequence of triangulations of $\cal P$ in which $T_{0} = T_{start}$ and $T_{r} = T_{end}$.  If $T_{i-1}$ is an underlying triangulation of $f_{i}$, and $T_{i}$ is a resulting triangulation of $f_{i}$ for each $i = 1,2,...,r$, we say $F$ transforms $T_{start}$ into $T_{end}$, or $F$ is a \emph{valid sequence}, denoted by $T_{start}\xrightarrow{F} T_{end}$.  The flip distance between $T_{start}$ and $T_{end}$ is the length of a shortest valid flip sequence.

Now we give the formal definition of \textsc{Parameterized Flip Distance} problem.

\begin{quote}
\textsc{Parameterized Flip Distance}\\
\textbf{Input:} Two triangulations $T_{start}$ and $T_{end}$ of $\cal P$ and an integer k.  \\
\textbf{Question:} Decide if the flip distance between $T_{start}$ and $T_{end}$ is equal to $k$.
\end{quote}

The triangulation on which we are performing a flip operation is called the \emph{current triangulation}.  An edge $e$ which belongs to the current triangulation but does not belong to $T_{end}$ is called a \emph{necessary edge} in the current triangulation.  It is easy to see that for any necessary edge $e$, there must exist a flip operation $f$ in a valid sequence such that $e$ = $\varepsilon (f)$.  Otherwise, we cannot get the objective triangulation $T_{end}$.

For a directed graph $D$, a maximal connected component of its underlying graph is called a \emph{weakly connected component} of $D$. We define the size of an undirected tree as the number of its vertices. A node in $D$ is called a \emph{source node} if the indegree of this node is $0$.

A \emph{parameterized problem} is a decision problem for which every instance is of the form $(x,k)$, where $x$ is the input instance and $k\in \mathbb{N}$ is the parameter.  A parameterized problem is \emph{fixed-parameter tractable} (\emph{FPT}) if it can be solved by an algorithm (\emph{FPT algorithm}) in $O(f(k)|x|^{O(1)})$ time, where $f(k)$ is a computable function of $k$. For a further introduction to parameterized algorithms, readers could refer to \cite{ChenSurvey,Cygan}.
\section{The Improved Algorithm for \textsc{Parameterized Flip Distance}}
\label{algorithm}
Given $T_{start}$ and $T_{end}$, let $F = \langle f_{1}, f_{2},...,f_{r}\rangle$ be a valid sequence, that is, $T_{start}\xrightarrow{F} T_{end}$. Definition \ref{def1} defines the adjacency of two flips in $F$.

\begin{definition} \label{def1}  \cite{Kanj}
Let $f_{i}$ and $f_{j}$ be two flips in $F$ ($1\leq i<j\leq r)$.  We define that flip $f_{j}$ is adjacent to flip $f_{i}$, denoted by $f_{i}\rightarrow f_{j}$, if the following two conditions are satisfied: \\
(1) either $\varphi (f_{i}) = \varepsilon (f_{j})$, or $ \varphi (f_{i})$ and $\varepsilon (f_{j})$ share a triangle in triangulation $T_{j-1}$; \\
(2) $\varphi (f_{i})$ is not flipped between $f_{i}$ and $f_{j}$, that is, there does not exist a flip $f_{p}$ in $F$, where \\
\hspace*{4mm} $i < p < j$, such that $\varphi (f_{i}) = \varepsilon (f_{p})$.
\end{definition}

By Definition \ref{def1}, we can construct a directed acyclic graph (DAG), denoted by $D_{F}$. Every node in $D_{F}$ represents a flip operation of $F$, and there is an arc from $f_{i}$ to $f_{j}$ if $f_{j}$ is adjacent to $f_{i}$.  For convenience, we label the nodes in $D_{F}$ using labels of the corresponding flip operations. In other words, we can see a node in $D_{F}$ as a flip operation and vice versa.

The intuition of Definition \ref{def1} is that if there is an arc from $f_{i}$ to $f_{j}$, then $f_j$ cannot be flipped before $f_i$ because the quadrilateral corresponding to the flip $f_j$ is formed after $f_i$ or the underlying edge of $f_j$, namely $\varepsilon (f_j)$ is the resulting edge of $f_i$, namely $\varphi (f_i)$.
The following lemma gives a stronger statement: any topological sorting of $D_{F}$ is a valid sequence.
\begin{lemma} \label{topo}  \cite{Kanj}
Let $T_{0}$ and $T_{r}$ be two triangulations and $F = \langle f_{1}, f_{2},...,f_{r}\rangle$ be a sequence of flips such that $T_{0}\xrightarrow{F} T_{r}$.  Let $\pi (F)$ be a permutation of the flips in $F$ such that $\pi (F)$ is a topological sorting of $D_{F}$.  Then $\pi (F)$ is a valid sequence of flips such that $T_{0}\xrightarrow{\pi (F)} T_{r}$.
\end{lemma}

Lemma \ref{topo} ensures that if we repeatedly remove a source node from $D_{F}$ and flip the underlying edge of this node until $D_{F}$ becomes empty, we can get a valid sequence and the objective triangulation $T_{end}$. On the basis of Lemma \ref{topo}, the essential task of our algorithm is to find an edge which is the underlying edge of a source node. Thus we introduce the definition of a \emph{walk}, which describes the "track" to find such an edge.

\begin{definition} \label{def2}  \cite{Kanj2}
A \emph{walk} in a triangulation $T$ (starting from an edge $e\in T$) is a sequence of edges of $T$ beginning with $e$ in which any two consecutive edges share a triangle in $T$.
\end{definition}

According to Lemma~\ref{topo}, if there is a valid sequence $F$ for the input instance, any topological sorting of $D_{F}$ is also a valid sequence for the given instance.  The difficulty is that $F$ is unknown. In order to find the topological sorting of $D_{F}$, the algorithm of Kanj and Xia \cite{Kanj} takes a nondeterministic walk to find an edge $e$ which is the underlying edge of a source node, flips this edge (removing the corresponding node from $D_{F}$), nondeterministically walks to an edge which shares a triangle with $e$ and recursively searches for an edge corresponding to a source node. Their algorithm deals with weakly connected components of $D_{F}$ one after another (refer to Corollary 4 in \cite{Kanj}), that is, the algorithm tries to find a solution $F$ in which all flips belonging to the same weakly connected component of $D_{F}$ appear consecutively.  In order to keep this procedure within the current weakly connected component, the algorithm uses a stack to preserve the nodes (defined as \emph{connecting point} in \cite{Kanj}) whose removal separates the current weakly connected component into small weakly connected components.  When removing all nodes of a small component, their algorithm jumps to the connecting point at the top of the stack in order to find another small component.

We observe that it is not necessary to remove all nodes of a weakly connected component before dealing with other weakly connected components, that is, our algorithm may find a solution $F$ in which the nodes belonging to the same weakly connected components appear dispersedly.  Thus our algorithm leaves out the stack which is used to preserve connecting points.  We show that it suffices to use two types of actions (Section~\ref{action}) instead of the five types in \cite{Kanj}.  Moreover, every time our algorithm finds a source node, it removes the node, flips the underlying edge and backtracks to the previous edge in the walk instead of searching for the next node, thus reducing the number of choices for the actions. In a word, our algorithm traverses $D_{F}$ in a reverse way. However, two adjacent edges in a walk may not correspond to two adjacent vertices in $D_{F}$. In order to emphasize this fact and make it more convenient for the proof, we construct an auxiliary graph $G$ and prove that $G$ is a forest. Actually $G$ can be seen as the track of the reverse traversal of $D_F$. However, $G$ is not the underlying graph of $D_{F}$ (Fig.~\ref{fig2}). Since there is a bijection between nondeterministic actions and nodes as well as edges of $G$, we prove that there exists a sequence of actions of length at most $2|D_{F}|$,  which is smaller than $11|D_{F}|$ in \cite{Kanj}.  In addition, we make some optimization on the strategy of finding the objective sequence.  As a result, we improve the running time of the algorithm from $O^{*}(k\cdot c^{k})$ where $c\leq 2\cdot 14^{11}$~\cite{Kanj} to $O^{*}(k\cdot 32^{k})$.

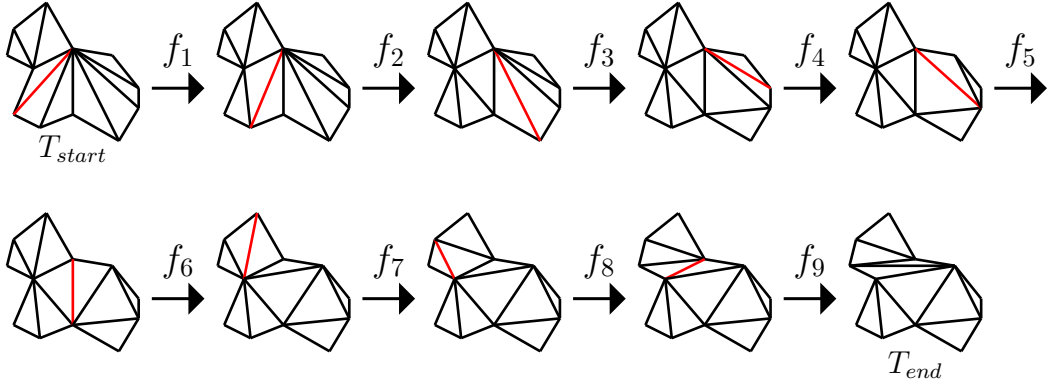
\begin{figure}[htbp]
\begin{center}
\begin{tikzpicture}[scale=0.35]

\begin{scope}[shift={(0cm,0cm)}]
		\node  (0) at (-4, 1) {};
		\node  (1) at (-5, 2.75) {};
		\node (2) at (-6.25, 1.75) {};
		\node (3) at (-6.5, 0.75) {};
		\node  (4) at (-5.5, 0.25) {};
		\node  (5) at (-6.25, -1.5) {};
		\node (6) at (-5.25, -2) {};
		\node (7) at (-4, -1.5) {};
		\node  (8) at (-2.25, -2.5) {};
		\node (9) at (-1.5, -1.25) {};
		\node  (10) at (-1.5, -0.5) {};
		\node (11) at (-2.5, 0.75) {};
        \node (12) at (-4, -2.8) {\large $T_{start}$};
		\node (13) at (-1, -0.5) {};
		\node (14) at (1, -0.5) {};

		\draw[line width=1] (1.center) to (0.center);
		\draw[line width=1] (1.center) to (2.center);
		\draw[line width=1] (2.center) to (3.center);
		\draw[line width=1] (3.center) to (4.center);
		\draw[line width=1] (4.center) to (5.center);
		\draw[line width=1] (5.center) to (6.center);
		\draw[line width=1] (6.center) to (7.center);
		\draw[line width=1] (7.center) to (8.center);
		\draw[line width=1] (8.center) to (9.center);
		\draw[line width=1] (9.center) to (10.center);
		\draw[line width=1] (10.center) to (11.center);
		\draw[line width=1] (11.center) to (0.center);
		\draw[line width=1] (2.center) to (4.center);
		\draw[line width=1] (1.center) to (4.center);
		\draw[line width=1] (4.center) to (0.center);
		\draw[line width=1,color=red] (0.center) to (5.center);
		\draw[line width=1] (0.center) to (6.center);
		\draw[line width=1] (0.center) to (7.center);
		\draw[line width=1] (0.center) to (8.center);
		\draw[line width=1] (0.center) to (9.center);
		\draw[line width=1] (0.center) to (10.center);
        \tikzset{>=triangle 90}
        \draw [->,line width=0.4mm] (13.center)--  node[xshift=0mm, yshift=5mm] {\Large $f_{1}$} (14.center);
\end{scope}

\begin{scope}[shift={(8cm,0cm)}]
		\node  (0) at (-4, 1) {};
		\node  (1) at (-5, 2.75) {};
		\node (2) at (-6.25, 1.75) {};
		\node (3) at (-6.5, 0.75) {};
		\node  (4) at (-5.5, 0.25) {};
		\node  (5) at (-6.25, -1.5) {};
		\node (6) at (-5.25, -2) {};
		\node (7) at (-4, -1.5) {};
		\node  (8) at (-2.25, -2.5) {};
		\node (9) at (-1.5, -1.25) {};
		\node  (10) at (-1.5, -0.5) {};
		\node (11) at (-2.5, 0.75) {};
		\node (13) at (-1, -0.5) {};
		\node (14) at (1, -0.5) {};

		\draw[line width=1] (1.center) to (0.center);
		\draw[line width=1] (1.center) to (2.center);
		\draw[line width=1] (2.center) to (3.center);
		\draw[line width=1] (3.center) to (4.center);
		\draw[line width=1] (4.center) to (5.center);
		\draw[line width=1] (5.center) to (6.center);
		\draw[line width=1] (6.center) to (7.center);
		\draw[line width=1] (7.center) to (8.center);
		\draw[line width=1] (8.center) to (9.center);
		\draw[line width=1] (9.center) to (10.center);
		\draw[line width=1] (10.center) to (11.center);
		\draw[line width=1] (11.center) to (0.center);
		\draw[line width=1] (2.center) to (4.center);
		\draw[line width=1] (1.center) to (4.center);
		\draw[line width=1] (4.center) to (0.center);
		\draw[line width=1] (4.center) to (6.center);
		\draw[line width=1,color=red] (0.center) to (6.center);
		\draw[line width=1] (0.center) to (7.center);
		\draw[line width=1] (0.center) to (8.center);
		\draw[line width=1] (0.center) to (9.center);
		\draw[line width=1] (0.center) to (10.center);
        \tikzset{>=triangle 90}
        \draw [->,line width=0.4mm] (13.center)--  node[xshift=0mm, yshift=5mm] {\Large $f_{2}$} (14.center);
\end{scope}

\begin{scope}[shift={(16cm,0cm)}]
		\node  (0) at (-4, 1) {};
		\node  (1) at (-5, 2.75) {};
		\node (2) at (-6.25, 1.75) {};
		\node (3) at (-6.5, 0.75) {};
		\node  (4) at (-5.5, 0.25) {};
		\node  (5) at (-6.25, -1.5) {};
		\node (6) at (-5.25, -2) {};
		\node (7) at (-4, -1.5) {};
		\node  (8) at (-2.25, -2.5) {};
		\node (9) at (-1.5, -1.25) {};
		\node  (10) at (-1.5, -0.5) {};
		\node (11) at (-2.5, 0.75) {};
		\node (13) at (-1, -0.5) {};
		\node (14) at (1, -0.5) {};

		\draw[line width=1] (1.center) to (0.center);
		\draw[line width=1] (1.center) to (2.center);
		\draw[line width=1] (2.center) to (3.center);
		\draw[line width=1] (3.center) to (4.center);
		\draw[line width=1] (4.center) to (5.center);
		\draw[line width=1] (5.center) to (6.center);
		\draw[line width=1] (6.center) to (7.center);
		\draw[line width=1] (7.center) to (8.center);
		\draw[line width=1] (8.center) to (9.center);
		\draw[line width=1] (9.center) to (10.center);
		\draw[line width=1] (10.center) to (11.center);
		\draw[line width=1] (11.center) to (0.center);
		\draw[line width=1] (2.center) to (4.center);
		\draw[line width=1] (1.center) to (4.center);
		\draw[line width=1] (4.center) to (0.center);
		\draw[line width=1] (4.center) to (6.center);
		\draw[line width=1] (4.center) to (7.center);
		\draw[line width=1] (0.center) to (7.center);
		\draw[line width=1,color=red] (0.center) to (8.center);
		\draw[line width=1] (0.center) to (9.center);
		\draw[line width=1] (0.center) to (10.center);
        \tikzset{>=triangle 90}
        \draw [->,line width=0.4mm] (13.center)--  node[xshift=0mm, yshift=5mm] {\Large $f_{3}$} (14.center);
\end{scope}

\begin{scope}[shift={(24cm,0cm)}]
		\node  (0) at (-4, 1) {};
		\node  (1) at (-5, 2.75) {};
		\node (2) at (-6.25, 1.75) {};
		\node (3) at (-6.5, 0.75) {};
		\node  (4) at (-5.5, 0.25) {};
		\node  (5) at (-6.25, -1.5) {};
		\node (6) at (-5.25, -2) {};
		\node (7) at (-4, -1.5) {};
		\node  (8) at (-2.25, -2.5) {};
		\node (9) at (-1.5, -1.25) {};
		\node  (10) at (-1.5, -0.5) {};
		\node (11) at (-2.5, 0.75) {};
		\node (13) at (-1, -0.5) {};
		\node (14) at (1, -0.5) {};

		\draw[line width=1] (1.center) to (0.center);
		\draw[line width=1] (1.center) to (2.center);
		\draw[line width=1] (2.center) to (3.center);
		\draw[line width=1] (3.center) to (4.center);
		\draw[line width=1] (4.center) to (5.center);
		\draw[line width=1] (5.center) to (6.center);
		\draw[line width=1] (6.center) to (7.center);
		\draw[line width=1] (7.center) to (8.center);
		\draw[line width=1] (8.center) to (9.center);
		\draw[line width=1] (9.center) to (10.center);
		\draw[line width=1] (10.center) to (11.center);
		\draw[line width=1] (11.center) to (0.center);
		\draw[line width=1] (2.center) to (4.center);
		\draw[line width=1] (1.center) to (4.center);
		\draw[line width=1] (4.center) to (0.center);
		\draw[line width=1] (4.center) to (6.center);
		\draw[line width=1] (4.center) to (7.center);
		\draw[line width=1] (0.center) to (7.center);
		\draw[line width=1] (7.center) to (9.center);
		\draw[line width=1] (0.center) to (9.center);
		\draw[line width=1,color=red] (0.center) to (10.center);
        \tikzset{>=triangle 90}
        \draw [->,line width=0.4mm] (13.center)--  node[xshift=0mm, yshift=5mm] {\Large $f_{4}$} (14.center);
\end{scope}

\begin{scope}[shift={(32cm,0cm)}]
		\node  (0) at (-4, 1) {};
		\node  (1) at (-5, 2.75) {};
		\node (2) at (-6.25, 1.75) {};
		\node (3) at (-6.5, 0.75) {};
		\node  (4) at (-5.5, 0.25) {};
		\node  (5) at (-6.25, -1.5) {};
		\node (6) at (-5.25, -2) {};
		\node (7) at (-4, -1.5) {};
		\node  (8) at (-2.25, -2.5) {};
		\node (9) at (-1.5, -1.25) {};
		\node  (10) at (-1.5, -0.5) {};
		\node (11) at (-2.5, 0.75) {};
		\node (13) at (-1, -0.5) {};
		\node (14) at (1, -0.5) {};

		\draw[line width=1] (1.center) to (0.center);
		\draw[line width=1] (1.center) to (2.center);
		\draw[line width=1] (2.center) to (3.center);
		\draw[line width=1] (3.center) to (4.center);
		\draw[line width=1] (4.center) to (5.center);
		\draw[line width=1] (5.center) to (6.center);
		\draw[line width=1] (6.center) to (7.center);
		\draw[line width=1] (7.center) to (8.center);
		\draw[line width=1] (8.center) to (9.center);
		\draw[line width=1] (9.center) to (10.center);
		\draw[line width=1] (10.center) to (11.center);
		\draw[line width=1] (11.center) to (0.center);
		\draw[line width=1] (2.center) to (4.center);
		\draw[line width=1] (1.center) to (4.center);
		\draw[line width=1] (4.center) to (0.center);
		\draw[line width=1] (4.center) to (6.center);
		\draw[line width=1] (4.center) to (7.center);
		\draw[line width=1] (0.center) to (7.center);
		\draw[line width=1] (7.center) to (9.center);
		\draw[line width=1,color=red] (0.center) to (9.center);
		\draw[line width=1] (9.center) to (11.center);
        \tikzset{>=triangle 90}
        \draw [->,line width=0.4mm] (13.center)--  node[xshift=0mm, yshift=5mm] {\Large $f_{5}$} (14.center);
\end{scope}

\begin{scope}[shift={(0,-8cm)}]
		\node  (0) at (-4, 1) {};
		\node  (1) at (-5, 2.75) {};
		\node (2) at (-6.25, 1.75) {};
		\node (3) at (-6.5, 0.75) {};
		\node  (4) at (-5.5, 0.25) {};
		\node  (5) at (-6.25, -1.5) {};
		\node (6) at (-5.25, -2) {};
		\node (7) at (-4, -1.5) {};
		\node  (8) at (-2.25, -2.5) {};
		\node (9) at (-1.5, -1.25) {};
		\node  (10) at (-1.5, -0.5) {};
		\node (11) at (-2.5, 0.75) {};
		\node (13) at (-1, -0.5) {};
		\node (14) at (1, -0.5) {};

		\draw[line width=1] (1.center) to (0.center);
		\draw[line width=1] (1.center) to (2.center);
		\draw[line width=1] (2.center) to (3.center);
		\draw[line width=1] (3.center) to (4.center);
		\draw[line width=1] (4.center) to (5.center);
		\draw[line width=1] (5.center) to (6.center);
		\draw[line width=1] (6.center) to (7.center);
		\draw[line width=1] (7.center) to (8.center);
		\draw[line width=1] (8.center) to (9.center);
		\draw[line width=1] (9.center) to (10.center);
		\draw[line width=1] (10.center) to (11.center);
		\draw[line width=1] (11.center) to (0.center);
		\draw[line width=1] (2.center) to (4.center);
		\draw[line width=1] (1.center) to (4.center);
		\draw[line width=1] (4.center) to (0.center);
		\draw[line width=1] (4.center) to (6.center);
		\draw[line width=1] (4.center) to (7.center);
		\draw[line width=1,color=red] (0.center) to (7.center);
		\draw[line width=1] (7.center) to (9.center);
		\draw[line width=1] (7.center) to (11.center);
		\draw[line width=1] (9.center) to (11.center);
        \tikzset{>=triangle 90}
        \draw [->,line width=0.4mm] (13.center)--  node[xshift=0mm, yshift=5mm] {\Large $f_{6}$} (14.center);
\end{scope}

\begin{scope}[shift={(8cm,-8cm)}]
		\node  (0) at (-4, 1) {};
		\node  (1) at (-5, 2.75) {};
		\node (2) at (-6.25, 1.75) {};
		\node (3) at (-6.5, 0.75) {};
		\node  (4) at (-5.5, 0.25) {};
		\node  (5) at (-6.25, -1.5) {};
		\node (6) at (-5.25, -2) {};
		\node (7) at (-4, -1.5) {};
		\node  (8) at (-2.25, -2.5) {};
		\node (9) at (-1.5, -1.25) {};
		\node  (10) at (-1.5, -0.5) {};
		\node (11) at (-2.5, 0.75) {};
		\node (13) at (-1, -0.5) {};
		\node (14) at (1, -0.5) {};

		\draw[line width=1] (1.center) to (0.center);
		\draw[line width=1] (1.center) to (2.center);
		\draw[line width=1] (2.center) to (3.center);
		\draw[line width=1] (3.center) to (4.center);
		\draw[line width=1] (4.center) to (5.center);
		\draw[line width=1] (5.center) to (6.center);
		\draw[line width=1] (6.center) to (7.center);
		\draw[line width=1] (7.center) to (8.center);
		\draw[line width=1] (8.center) to (9.center);
		\draw[line width=1] (9.center) to (10.center);
		\draw[line width=1] (10.center) to (11.center);
		\draw[line width=1] (11.center) to (0.center);
		\draw[line width=1] (2.center) to (4.center);
		\draw[line width=1,color=red] (1.center) to (4.center);
		\draw[line width=1] (4.center) to (0.center);
		\draw[line width=1] (4.center) to (6.center);
		\draw[line width=1] (4.center) to (7.center);
		\draw[line width=1] (4.center) to (11.center);
		\draw[line width=1] (7.center) to (9.center);
		\draw[line width=1] (7.center) to (11.center);
		\draw[line width=1] (9.center) to (11.center);
        \tikzset{>=triangle 90}
        \draw [->,line width=0.4mm] (13.center)--  node[xshift=0mm, yshift=5mm] {\Large $f_{7}$} (14.center);
\end{scope}

\begin{scope}[shift={(16cm,-8cm)}]
		\node  (0) at (-4, 1) {};
		\node  (1) at (-5, 2.75) {};
		\node (2) at (-6.25, 1.75) {};
		\node (3) at (-6.5, 0.75) {};
		\node  (4) at (-5.5, 0.25) {};
		\node  (5) at (-6.25, -1.5) {};
		\node (6) at (-5.25, -2) {};
		\node (7) at (-4, -1.5) {};
		\node  (8) at (-2.25, -2.5) {};
		\node (9) at (-1.5, -1.25) {};
		\node  (10) at (-1.5, -0.5) {};
		\node (11) at (-2.5, 0.75) {};
		\node (13) at (-1, -0.5) {};
		\node (14) at (1, -0.5) {};

		\draw[line width=1] (1.center) to (0.center);
		\draw[line width=1] (1.center) to (2.center);
		\draw[line width=1] (2.center) to (3.center);
		\draw[line width=1] (3.center) to (4.center);
		\draw[line width=1] (4.center) to (5.center);
		\draw[line width=1] (5.center) to (6.center);
		\draw[line width=1] (6.center) to (7.center);
		\draw[line width=1] (7.center) to (8.center);
		\draw[line width=1] (8.center) to (9.center);
		\draw[line width=1] (9.center) to (10.center);
		\draw[line width=1] (10.center) to (11.center);
		\draw[line width=1] (11.center) to (0.center);
		\draw[line width=1,color=red] (2.center) to (4.center);
		\draw[line width=1] (0.center) to (2.center);
		\draw[line width=1] (4.center) to (0.center);
		\draw[line width=1] (4.center) to (6.center);
		\draw[line width=1] (4.center) to (7.center);
		\draw[line width=1] (4.center) to (11.center);
		\draw[line width=1] (7.center) to (9.center);
		\draw[line width=1] (7.center) to (11.center);
		\draw[line width=1] (9.center) to (11.center);
        \tikzset{>=triangle 90}
        \draw [->,line width=0.4mm] (13.center)--  node[xshift=0mm, yshift=5mm] {\Large $f_{8}$} (14.center);
\end{scope}

\begin{scope}[shift={(24cm,-8cm)}]
		\node  (0) at (-4, 1) {};
		\node  (1) at (-5, 2.75) {};
		\node (2) at (-6.25, 1.75) {};
		\node (3) at (-6.5, 0.75) {};
		\node  (4) at (-5.5, 0.25) {};
		\node  (5) at (-6.25, -1.5) {};
		\node (6) at (-5.25, -2) {};
		\node (7) at (-4, -1.5) {};
		\node  (8) at (-2.25, -2.5) {};
		\node (9) at (-1.5, -1.25) {};
		\node  (10) at (-1.5, -0.5) {};
		\node (11) at (-2.5, 0.75) {};
		\node (13) at (-1, -0.5) {};
		\node (14) at (1, -0.5) {};

		\draw[line width=1] (1.center) to (0.center);
		\draw[line width=1] (1.center) to (2.center);
		\draw[line width=1] (2.center) to (3.center);
		\draw[line width=1] (3.center) to (4.center);
		\draw[line width=1] (4.center) to (5.center);
		\draw[line width=1] (5.center) to (6.center);
		\draw[line width=1] (6.center) to (7.center);
		\draw[line width=1] (7.center) to (8.center);
		\draw[line width=1] (8.center) to (9.center);
		\draw[line width=1] (9.center) to (10.center);
		\draw[line width=1] (10.center) to (11.center);
		\draw[line width=1] (11.center) to (0.center);
		\draw[line width=1] (0.center) to (3.center);
		\draw[line width=1] (0.center) to (2.center);
		\draw[line width=1,color=red] (4.center) to (0.center);
		\draw[line width=1] (4.center) to (6.center);
		\draw[line width=1] (4.center) to (7.center);
		\draw[line width=1] (4.center) to (11.center);
		\draw[line width=1] (7.center) to (9.center);
		\draw[line width=1] (7.center) to (11.center);
		\draw[line width=1] (9.center) to (11.center);
        \tikzset{>=triangle 90}
        \draw [->,line width=0.4mm] (13.center)--  node[xshift=0mm, yshift=5mm] {\Large $f_{9}$} (14.center);
\end{scope}

\begin{scope}[shift={(32cm,-8cm)}]
		\node  (0) at (-4, 1) {};
		\node  (1) at (-5, 2.75) {};
		\node (2) at (-6.25, 1.75) {};
		\node (3) at (-6.5, 0.75) {};
		\node  (4) at (-5.5, 0.25) {};
		\node  (5) at (-6.25, -1.5) {};
		\node (6) at (-5.25, -2) {};
		\node (7) at (-4, -1.5) {};
		\node  (8) at (-2.25, -2.5) {};
		\node (9) at (-1.5, -1.25) {};
		\node  (10) at (-1.5, -0.5) {};
		\node (11) at (-2.5, 0.75) {};
        \node (12) at (-4, -3) {\large $T_{end}$};

		\draw[line width=1] (1.center) to (0.center);
		\draw[line width=1] (1.center) to (2.center);
		\draw[line width=1] (2.center) to (3.center);
		\draw[line width=1] (3.center) to (4.center);
		\draw[line width=1] (4.center) to (5.center);
		\draw[line width=1] (5.center) to (6.center);
		\draw[line width=1] (6.center) to (7.center);
		\draw[line width=1] (7.center) to (8.center);
		\draw[line width=1] (8.center) to (9.center);
		\draw[line width=1] (9.center) to (10.center);
		\draw[line width=1] (10.center) to (11.center);
		\draw[line width=1] (11.center) to (0.center);
		\draw[line width=1] (0.center) to (3.center);
		\draw[line width=1] (0.center) to (2.center);
		\draw[line width=1] (3.center) to (11.center);
		\draw[line width=1] (4.center) to (6.center);
		\draw[line width=1] (4.center) to (7.center);
		\draw[line width=1] (4.center) to (11.center);
		\draw[line width=1] (7.center) to (9.center);
		\draw[line width=1] (7.center) to (11.center);
		\draw[line width=1] (9.center) to (11.center);
\end{scope}
\end{tikzpicture}
\end{center}
\caption{An example of \textsc{Flip Distance}. $F=\langle f_{1},...,f_{9}\rangle$ is a shortest valid sequence for $T_{start}$ and $T_{end}$. The red edge is the current edge to be flipped.}
\label{fig4}
\end{figure}
\subsection{Nondeterministic construction process}
Now we give a description of our nondeterministic construction process \textbf{NDTRV} (see Fig.~\ref{fig1}).  The construction is nondeterministic, that is, we suppose it always guesses the optimal choice correctly when running.  The actual deterministic algorithm enumerates all possible choices to simulate the nondeterministic actions (see Fig.~\ref{fig3}). Readers could refer to \cite{Jianer} as an example of nondeterministic algorithm.  We present this construction process in order to depict the idea behind our deterministic algorithm clearly and vividly.

Let $T_{start}$ be the initial triangulation, and $T_{end}$ be the objective triangulation.  Suppose that $F$ is a shortest valid sequence, that is, $F$ has the shortest length among all valid sequences.  Let $D_{F}$ be the DAG constructed after $F$ according to Definition~\ref{def1}.
\textbf{NDTRV} traverses  $D_{F}$ reversely, removes the vertices of $D_{F}$ in a topologically-sorted order and transforms $T_{start}$ into $T_{end}$.  Although $D_{F}$ is unknown, for further analysis, we assume that \textbf{NDTRV} can remove and copy nodes in $D_{F}$ so that it can construct an auxiliary undirected graph $G$ and a list $L$ during the traversal.
In later analysis we show that $G$ is a forest. Moreover, there is a bijection between flipping actions of \textbf{NDTRV} and nodes of $G$ while there is a bijection between moving actions of \textbf{NDTRV} and edges of $G$. Obviously $G$ and $L$ are unknown as well.  We just show that if a shortest valid sequence $F$ exists, then $D_{F}$ exists.  So do $G$ and $L$.  We can see $D_{F}$ and $G$ as conceptual or dummy graphs.  We construct $G$ instead of analysing a subgraph of $D_{F}$ because one moving action (see Section~\ref{action}) of \textbf{NDTRV} may correspond to one or more edges in $D_{F}$ (see Fig.~\ref{fig2}), while there is a one-to-one correspondence between moving actions and edges in $G$ .

At the beginning of an iteration, \textbf{NDTRV} picks a necessary edge $e=\varepsilon (f_{h})$ arbitrarily and nondeterministically guesses a walk $W$ to find the underlying edge of a source node $f_{s}$. Lemma~\ref{backbone} shows that there exists such a walk $W$ whose length is bounded by the length of a directed path $B$ from $f_{s}$ to $f_{h}$, and every edge $e'$ in $W$ is the underlying edge of some flip $f'$ on $B$. \textbf{NDTRV} uses $L$ to preserve a sequence of nodes $\Gamma=\langle f_{s}=v_{1},...,f_{h}=v_{\ell}\rangle$ on $B$, whose underlying edges are in $W$. Simultaneously \textbf{NDTRV} constructs a path $S$ by copying all nodes in $\Gamma$ as well as adding an undirected edge between the copy of $v_{i}$ and $v_{i+1}$ for $i=1,...,\ell$.  $S$ is defined as a \emph{searching path}.  The node $f_{h}$ is called a \emph{starting node}.  If a starting node is precisely a source node in $D_{F}$, the searching path consists only of the copy of this starting node.  When finding $\varepsilon (f_{s})$, \textbf{NDTRV} removes $f_{s}$ from $D_{F}$, flips $\varepsilon (f_{s})$ and moves back(backtracks) to the previous edge $\varepsilon (v_{2})$ of $\varepsilon (f_{s})$ in $W$. If $v_{2}$ becomes a source node of $D_{F}$, \textbf{NDTRV} removes $v_{2}$ from $D_{F}$, flips $\varepsilon (v_{2})$ and moves back to the previous edge $\varepsilon (v_{3})$.  \textbf{NDTRV} repeats the above operations until finding a node $v_{i}$ in $\Gamma$ which is not a source node in $D_{F}$.  Then \textbf{NDTRV} uses $v_{i}$ as a new starting node, and recursively guesses a walk nondeterministically from $\varepsilon (v_{i})$ to find another edge which is the underlying edge of a source node as above.  \textbf{NDTRV} performs these operations until the initial starting node $f_{h}$ becomes a source node in $D_{F}$.  Finally \textbf{NDTRV} removes $f_{h}$ and flips $\varepsilon (f_{h})$, terminating this iteration.  If $T_{current}$ is not equal to $T_{end}$, \textbf{NDTRV} picks a new necessary edge and starts a new iteration as above until $T_{start}$ is transformed into $T_{end}$.  We give the formal presentation of \textbf{NDTRV} in Fig.~\ref{fig1} and an example in Fig.~\ref{fig4} and Fig.~\ref{fig2}.

\subsection{Actions of the construction}   \label{action}
Our construction process contains two types of actions operating on triangulations. The edge which the algorithm is operating on is called \emph{the current edge}. The current triangulation is denoted by $T_{current}$.

\begin{quote}
(i) Move to one edge that shares a triangle with the current edge in $T_{current}$.  We formalize it as \\
\hspace*{3mm} $(move,e_{1}\mapsto e_{2})$, where $e_{1}$ is the current edge and $e_{2}$ shares a triangle with $e_{1}$.

(ii) Flip the current edge and move back to the previous edge of the  current edge in $W$.  We \\
\hspace*{4mm} formalize it as $(f,e_{4}\mapsto e_{3})$, where $f$ is the flip performed on the current edge, $e_{4}$ equals $\varphi (f)$ \\
\hspace*{4mm} and $e_{3}$ is the previous edge of $\varepsilon (f)$ in the current walk $W$.
\end{quote}

Since there are four edges that share a triangle with the current edge, there are at most four directions for an action of type (i).  However, there is only one choice for an action of type (ii).

\subsection{The sequence of actions} \label{section4.3}
The following theorem is the main theorem for the deterministic algorithm \textbf{FLIPDT}, which bounds the length of the sequence of actions by $2|V(D_{F})|$.
\begin{theorem}  \label{main}
There exists a sequence of actions of length at most $2|V(D_{F})|$ following which we can perform a sequence of flips $F'$ of length $|V(D_{F})|$, starting from a necessary edge in $T_{start}$, such that $F'$ is a topological sorting of $D_{F}$.
\end{theorem}

In order to prove theorem \ref{main}, we need to introduce some lemmas.  We will give the proof for Theorem~\ref{main} at the end of Section~\ref{section4.3}.

Lemma~\ref{backbone} is one of the main structural results of Kanj and Xia \cite{Kanj2}. It is also the structural basis of our algorithm.

\begin{lemma} \label{backbone}  \cite{Kanj2}
Suppose that a sequence of flips $F^{-}$ is performed such that every time we flip an edge, we delete the corresponding source node in the DAG resulting from preceding deleting operations.  Let $f_{h}$ be a node in the remaining DAG such that $\varepsilon (f_{h})$ is an edge in the triangulation $T$ resulting from performing the sequence of flips $F^{-}$.  There is a source node $f_{s}$ in the remaining DAG satisfying:\\
(1) There is a walk $W$  in $T$ from $\varepsilon (f_{h})$ to $\varepsilon (f_{s})$.\\
(2) There is a directed path $B$ from $f_{s}$ to $f_{h}$ in the remaining DAG that we refer to as the \\
\hspace*{5mm} backbone of the DAG.\\
(3) The length of $W$ is at most that of $B$.\\
(4) Any edge in $W$ is the underlying edge of a flip in $B$, that is, $W=\langle \varepsilon(v_{1}),...,\varepsilon(v_{\ell})\rangle$, where\\
\hspace*{5mm}  $v_{1}=f_{s},...,v_{\ell}=f_{h}$ are nodes in $B$ and there is a directed path $B_{i}$ from $v_{i}$ to $v_{i+1}$ for\\
\hspace*{5mm}  $i=1,...,\ell-1$ such that $B_{i}\subset B$.
\end{lemma}

\vspace*{-5mm}

\begin{figure}[htbp]
\setbox4=\vbox{\hsize32pc \strut \begin{quote}
\vspace*{-5mm} \footnotesize
\textbf{\textbf{NDTRV}}($T_{start}$, $T_{end}$; $D_{F}$) \\
\hspace*{4mm} Input: the initial triangulation $T_{start}$ and objective triangulation $T_{end}$. \\

\hspace*{4mm} /*Assume that $F$ is a shortest sequence, $D_{F}$ is the corresponding DAG by Definition \ref{def1}*/\\
\hspace*{4mm} /*$G$ is an auxiliary undirected graph */ \\
\hspace*{4mm} /*$L$ is a list keeping track of searching paths for backtracking */ \\
\hspace*{4mm} /*$Q$ is a list preserving the sequence of nondeterministic actions */ \\
\hspace*{4mm} /* $T_{current}$ is the current triangulation*/ \\

\begin{adjustwidth}{4mm}{0mm}
\hspace*{2mm} a. Let $V(G)$ and $E(G)$ be empty sets, $L$ and $Q$ be empty lists;  \\
\hspace*{2mm} b. $T_{current}=T_{start}$;  \\
\hspace*{2mm} c. {\bf While} $T_{current}\neq T_{end}$ \textbf{do}\\
\hspace*{2mm} c.1.\hspace{2mm} Pick a necessary edge $e= \varepsilon (f_{h})$ in $T_{current}$ arbitrarily; \\
\hspace*{2mm} c.2.\hspace{2mm} Add a copy of $f_{h}$ to $G$; \\
\hspace*{2mm} c.3.\hspace{2mm} Append $f_{h}$ to $L$; \\
\hspace*{2mm} c.4.\hspace{2mm} \textbf{TrackTree}($T_{current}$, $e$, $D_{F}$, $G$, $L$, $Q$); \\

{\bf TrackTree}($T_{current}$, $\varepsilon (f_{h})$, $D_{F}$, $G$, $L$, $Q$)   \\
\hspace*{2mm} 1. Nondeterministically guess a walk in $T_{current}$ from $\varepsilon (f_{h})$ to find $\varepsilon (f_{s})$ according to \\
\hspace*{6mm} Lemma~\ref{backbone}, let $\Gamma=\langle f_{s}=v_{1},...,f_{h}=v_{\ell}\rangle$,  where $f_{s}$ is a source node in  $D_{F}$, be a sequence \\
\hspace*{6mm} of nodes on the backbone $B$ whose underlying  edges are in the walk $W$  such that $\varepsilon (v_{i})$ \\
\hspace*{6mm} and $\varepsilon (v_{i+1})$ are consecutive in $W$ for $i =1,...,\ell-1$; \\
\hspace*{2mm} 2. Add a copy of $v_{1}$,...,$v_{\ell-1}$ to $G$ respectively; \\
\hspace*{2mm} 3. Connect the copies of $v_{1}$,...,$v_{\ell}$ in $G$ into a path; \\
\hspace*{2mm} 4. Append $v_{\ell-1}$,...,$v_{1}$ to $L$; \hspace*{35mm} \\
\hspace*{2mm} 5. Append $(move, \varepsilon (v_{\ell})\mapsto \varepsilon (v_{\ell-1}))$,...,$(move, \varepsilon (v_{2})\mapsto \varepsilon (v_{1}))$ to $Q$;  \hspace*{2mm} \emph{/*record actions*/} \\
\hspace*{2mm} 6. Remove $f_{s}=v_{1}$ from $L$; \\
\hspace*{2mm} 7. Remove $f_{s}=v_{1}$ from $D_{F}$; \\
\hspace*{2mm} 8. Flip $\varepsilon (f_{s})$ in $T_{current}$ and move back to $\varepsilon (v_{2})$; \\
\hspace*{2mm} 9. Append ($f_{s}$, $\varphi (v_{1})\mapsto \varepsilon(v_{2}))$ to $Q$; \hspace*{6mm} \emph{/*record actions*/} \\
\hspace*{2mm}10. {\bf For} $i = 2$ \textbf{to} $\ell$ \textbf{do}\\
\hspace*{2mm}10.1\hspace{3mm}Nondeterministically guess if $v_{i}$ is a source node in $D_{F}$;  \\
\hspace*{2mm}10.2\hspace{3mm}{\bf If} $v_{i}$ is a source node of $D_{F}$ \textbf{then} \hspace{1mm} \emph{/*flip and move back*/}\\
\hspace*{2mm}10.2.1\hspace{3mm} Remove $v_{i}$ from $L$;\\
\hspace*{2mm}10.2.2\hspace{3mm} Remove $v_{i}$ from $D_{F}$; \\
\hspace*{2mm}10.2.3\hspace{3mm} Flip $\varepsilon (v_{i})$ in $T_{current}$ and move back to $\varepsilon (v_{i+1})$;\\
\hspace*{2mm}10.2.4\hspace{3mm} Append ($v_{i}$, $\varphi (v_{i})\mapsto \varepsilon(v_{i+1}))$ to $Q$; \\
\hspace*{2mm}10.3\hspace{3mm}{\bf Else} \hspace{15mm} \\
\hspace*{2mm}10.3.1\hspace{3mm} \textbf{TrackTree}($T_{current}$, $\varepsilon (v_{i})$, $D_{F}$, $G$, $L$, $Q$);
\end{adjustwidth}
\end{quote}

\vspace*{-6mm}
\strut}
$$\boxit{\box4}$$
\vspace*{-9mm}
\caption{Nondeterministic construction \textbf{NDTRV}} \label{fig1}
\end{figure}

\begin{lemma}  \label{correction}
\textbf{NDTRV} transforms $T_{start}$ into $T_{end}$ with the minimum number of flips and stops in polynomial time if it correctly guesses every moving and flipping action.
\end{lemma}
\begin{proof}
Suppose that $F$ is a shortest valid sequence.  According to Lemma \ref{backbone}, every edge flipped in \textbf{NDTRV} is the underlying edge of a source node in the remaining graph of $D_{F}$, and every node removed from the remaining graph of $D_{F}$ in \textbf{NDTRV} is a source node. If $T_{current}$ is equal to $T_{end}$ but $D_{F}$ is not empty, then there exists a valid sequence $F'$ which is shorter than $F$, contradicting that $F$ is a shortest valid sequence.  Thus \textbf{NDTRV} traverses $D_{F}$, removes all nodes of $D_{F}$ in a topologically-sorted order and transforms $T_{start}$ into $T_{end}$ with the minimum number of flips by Lemma \ref{topo}.  Since the diameter of a transformations graph $G_{T}(\cal P)$ is $O(n^{2})$ \cite{Lawson}, \textbf{NDTRV} stops in polynomial time.
\end{proof}

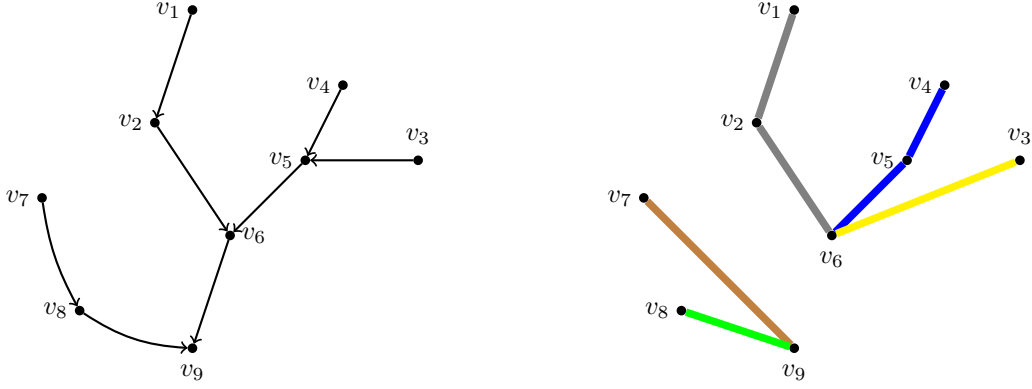
\begin{figure}[htbp]
\begin{center}
\begin{tikzpicture} [auto=left,every node/.style={fill=black,circle,draw,inner sep=1pt}, every path/.style={thick},scale=1.0]
\node [label=below:{$v_{9}$}](v9) at (0.5,-4) {};
\node [label=right:{$v_{6}$}](v6) at (1,-2.5) {};
\node [label=left:{$v_{2}$}](v2) at (0,-1) {};
\node [label=left:{$v_{1}$}](v1) at (0.5,0.5) {};
\node [label=left:{$v_{5}$}](v5) at (2,-1.5) {};
\node [label=left:{$v_{4}$}](v3) at (2.5,-0.5) {};
\node [label=above:{$v_{3}$}](v4) at (3.5,-1.5) {};
\node [label=left:{$v_{8}$}](v8) at (-1,-3.5) {};
\node [label=left:{$v_{7}$}](v7) at (-1.5,-2) {};
\draw[->] (v1) -- (v2);
\draw[->] (v2) -- (v6);
\draw[->] (v3) -- (v5);
\draw[->] (v4) -- (v5);
\draw[->] (v5) -- (v6);
\draw[->] (v6) -- (v9);
\draw[bend right=10,->] (v7) to (v8);
\draw[bend right=15,->] (v8) to (v9);

\begin{scope}[xshift=8cm]
\node [label=below:{$v_{9}$}](v9) at (0.5,-4) {};
\node [label=below:{$v_{6}$}](v6) at (1,-2.5) {};
\node [label=left:{$v_{2}$}](v2) at (0,-1) {};
\node [label=left:{$v_{1}$}](v1) at (0.5,0.5) {};
\node [label=left:{$v_{5}$}](v5) at (2,-1.5) {};
\node [label=left:{$v_{4}$}](v3) at (2.5,-0.5) {};
\node [label=above:{$v_{3}$}](v4) at (3.5,-1.5) {};
\node [label=left:{$v_{8}$}](v8) at (-1,-3.5) {};
\node [label=left:{$v_{7}$}](v7) at (-1.5,-2) {};
\draw[gray,solid,line width=1mm,fill=gray] (v1)  -- (v2) -- (v6);
\draw[blue,solid,line width=1mm,fill=blue] (v3) -- (v5) -- (v6);
\draw[yellow,solid,line width=1mm,fill=yellow]  (v6) -- (v4);
\draw[brown,solid,line width=1mm,fill=green] (v7) -- (v9);
\draw[green,solid,line width=1mm,fill=green] (v8) -- (v9);
\end{scope}
\end{tikzpicture}
\end{center}
\caption{Constructing $G$ according to the example in Fig.~\ref{fig4}. The graph on the left is $D_F$ after $F$ in Fig.~\ref{fig4}. The graph on the right is the auxiliary graph $G$ constructed by \textbf{NDTRV}. $\varepsilon (v_{6})$ is the necessary edge chosen at the beginning of the first iteration and $\varepsilon (v_{9})$ is the one chosen at the beginning of the second iteration. }
\label{fig2}
\end{figure}

\begin{lemma} \label{forest}
The auxiliary graph $G$ constructed during \textbf{NDTRV} is a forest and $|V(G)|=|V(D_{F})|$. Morover, $G$ consists of a set of vertex-disjoint trees called \emph{track trees}. Each track tree is created during an iteration of \textbf{NDTRV}.
\end{lemma}
\begin{proof}
Since \textbf{NDTRV} makes a topological sorting during execution and the copy of a vertex of $D_{F}$ is added to $G$ only when it is removed from $D_{F}$, it follows that $|V(G)|=|V(D_{F})|$. Suppose that there is a cycle in $G$. According to lemma~\ref{backbone} and \textbf{NDTRV}, we can find a directed cycle in $D_{F}$, contradicting that $D_{F}$ is a directed acyclic graph. Thus $G$ is a forest. From the execution of \textbf{NDTRV}, we get that it creates a connected subgraph in $G$ during every iteration and the subgraphs created during each iteration are vertex-disjoint.  Thus the subgraph created during each iteration is a tree. This concludes the proof.
\end{proof}

We give the proof of Theorem~\ref{main} below.
\begin{proof}(Theorem~\ref{main})
During the procedure of \textbf{NDTRV}(Fig.~\ref{fig1}), it constructs a list $Q$ consisting of actions of type (i) and (ii).  We claim that $Q$ is exactly the sequence satisfying the requirement of this theorem. \textbf{NDTRV} appends an action of type (ii) to $Q$ if and only if it adds a vertex to $G$. Meanwhile, \textbf{NDTRV} appends an action of type (i) to $Q$ if and only if it adds an an edge to $G$.   It follows that there is a one-to-one correspondence between actions of type (i) in $Q$ and $E(G)$ and there is a one-to-one correspondence between actions of type (ii) in $Q$ and $V(G)$. According to Lemma~\ref{forest}, $G$ is a forest.  As a result, $|E(G)|\leq |V(G)|$, and the length of $Q$ is bounded by $|E(G)|+|V(G)|\leq 2|V(G)|=2|V(D_{F})|$.
\end{proof}

\subsection{The deterministic algorithm}
Now we are ready to give the deterministic algorithm \textbf{FLIPDT} for \textsc{Parameterized Flip Distance}. The specific algorithm is presented in Fig.~\ref{fig3}. As mentioned above, we assume that \textbf{NDTRV} is always able to guess the optimal choice correctly.  In fact, \textbf{FLIPDT} achieves this by trying all possible sequences of actions and partitions of $k$.  At the top level, \textbf{FLIPDT} branches into all partitions of $k$, namely $(k_{1},...,k_{t})$ satisfying $k_{1}+...+k_{t}=k$ and $k_{1},...,k_{t}\geq 1$, in which $k_{i}$ ($i= 1,...,t$) equals the size of the track tree $A_{i}$ constructed during the $i$-th iteration.

Suppose that \textbf{FLIPDT} is under some partition $(k_{1},...,k_{t})$.  Let $T_{iteration}^{0}=T_{start}$.  \textbf{FLIPDT} permutates all necessary edges in $T_{start}$ in the lexicographical order, and the ordering is denoted by $O_{lex}$.  Here we number the given points of $\cal P$ in the Euclidean plane from $1$ to $n$ arbitrarily and label one edge by a tuple consisting of two numbers of its endpoints. Thus we can order the edges lexicographically.  \textbf{FLIPDT} performs $t$ iterations.  At the beginning of the $i$-th iteration, $i = 1,...,t$, we denote the current triangulation by $T_{iteration}^{i-1}$.  For $i=1,...,t$, $T_{iteration}^{i}$ is also the triangulation resulting from the execution of the first $i$ iterations. At the beginning of the $i$-th iteration ($i=1,...,t$), \textbf{FLIPDT} repeatedly picks the next edge in $O_{lex}$ until finding a necessary edge $e$ belonging to $T_{iteration}^{i-1}$.  Note that one edge in $O_{lex}$ may not be a necessary edge anymore with respect to $T_{iteration}^{i-1}$.  Moreover, if \textbf{FLIPDT} reaches the end of $O_{lex}$ but does not find a necessary edge belonging to $T_{iteration}^{i-1}$, it needs to update $O_{lex}$ by clearing $O_{lex}$ and permutating all necessary edges in $T_{iteration}^{i-1}$ lexicographically, and choose the first edge in the updated ordering $O_{lex}$. Then \textbf{FLIPDT} branches into every possible sequence of actions $seq_{i}$ of length $2k_{i}-1$.

Under each enumeration of $seq_{i}$, \textbf{FLIPDT} branches into every possible sequence of actions $seq_{i+1}$ of length $2k_{i+1}-1$.  \textbf{FLIPDT} proceeds as above.  When \textbf{FLIPDT} finishes the last iteration, it judges if the resulting triangulation $T_{iteration}^{t}$ is equal to $T_{end}$.  If they are equal, the input instance is a yes-instance.  Otherwise, \textbf{FLIPDT} rejects this branch and proceeds.

Now we analyse how to enumerate all possible sequences of length $2k_{i}-1$. According to Lemma~\ref{forest} and Theorem~\ref{main}, in every iteration \textbf{NDTRV} constructs a track tree in which a node corresponds to an action of type (ii) while an edge corresponds to an action of type (i).  It follows that the number of actions of type (ii) is $k_{i}$, and the number of actions of type (i) is $k_{i}-1$.  According to \textbf{NDTRV}, the last action $\gamma$ in $seq_{i}$ must be of type (ii), and in any prefix of $seq_{i}-\gamma$ the number of actions of type (i) must not be less than that of type (ii).  Thus \textbf{FLIPDT} only needs to enumerate all sequences of length $2k_{i}-1$ satisfying the above constraints.

\begin{figure}[htb]
\setbox4=\vbox{\hsize32pc \strut \begin{quote}
\vspace*{-5mm} \footnotesize
{\bf FLIPDT}$(T_{start},T_{end}, k)$\\
\hspace*{2mm} Input: two triangulations $T_{start}$ and $T_{end}$ of a point set $\cal P$ in the Euclidean plane and an \\
\hspace*{12mm} integer $k$. \\
\hspace*{2mm} Output: return YES if there exists a sequence of flips of length $k$ that transforms $T_{start}$  \\
\hspace*{14mm} into $T_{end}$; otherwise return NO. \\

\hspace*{2mm} 1. {\bf For} each partition $(k_{1},...,k_{t})$ of $k$ satisfying $k_{1}+k_{2}+...+k_{t}=k$ and $k_{1},...,k_{t}\geq 1$ \textbf{do}\\
\hspace*{2mm} 1.1\hspace{5mm}Order all necessary edges in $T_{start}$ lexicographically and denote this ordering by $O_{lex}$;  \\
\hspace*{2mm} 1.2\hspace{5mm}\textbf{FDSearch}($T_{start}$,1,$(k_{1},...,k_{t})$); \emph{/*iteration $1$ distributed with $k_{1}$*/} \\
\hspace*{2mm} 2. \textbf{Return} NO; \\

{\bf FDSearch}($T$,$i$,$(k_{1},...,k_{t})$)  \hspace*{2mm} \emph{/*the concrete branching procedure*/ }\\
\hspace*{2mm} Input: a triangulation $T$, an integer $i$ denoting that the algorithm is at the $i$-th iteration and \\
\hspace*{12mm} a partition $(k_{1},...,k_{t})$ of $k$.\\
\hspace*{2mm} Output: return YES if the instance is accepted. \\

\hspace*{2mm} 1.\hspace{0mm} Repeatedly pick the next edge in $O_{lex}$ until finding a necessary edge $e$ with respect to $T$ \\
\hspace*{6mm} and $T_{end}$;\\
\hspace*{2mm} 2. \hspace{0mm}{\bf If} it reaches the end of $O_{lex}$ but finds no necessary edge in $T$ \textbf{then}\\
\hspace*{2mm} 2.1\hspace{3mm} Update $O_{lex}$ by permutating all necessary edges in $T_{iteration}^{i-1}$ in lexicographical order, \\
\hspace*{10mm} and pick the first edge $e$ in $O_{lex}$;  \\
\hspace*{2mm} 3.\hspace{0mm} {\bf For} each possible sequence of actions $seq_{i}$ of length $2k_{i}-1$ \textbf{do} \\
\hspace*{2mm} 3.1\hspace{3mm} $T'$ = \textbf{Transform}($T$,$seq_{i}$,$e$); \\
\hspace*{2mm} 3.2\hspace{3mm} {\bf If} $i<t$ \textbf{then}\hspace*{2mm} \emph{/*continue to the next iteration distributed with $k_{i+1}$*/} \\
\hspace*{2mm} 3.2.1\hspace{6mm} \textbf{FDSearch}($T'$,$i+1$,$(k_{1},...,k_{t})$);  \\
\hspace*{2mm} 3.3\hspace{3mm} {\bf Else if} $i=t$ and $T'=T_{end}$ \textbf{then} \hspace*{2mm} \emph{/*compare $T'$ with $T_{end}$*/}\\
\hspace*{2mm} 3.3.1\hspace{6mm} \textbf{Return} YES;  \\

{\bf Transform}($T$,$s$,$e$) \hspace*{2mm} \emph{/*subprocess for transforming triangulations*/ }\\
\hspace*{2mm} Input:  a triangulation $T$, a sequence of actions $s$ and a starting edges $e$.\\
\hspace*{2mm} Output:  a new triangulation $T'$. \\

\hspace*{2mm} 1. Perform a sequence of actions $s$ starting from $e$ in $T$, getting a new triangulation $T'$;  \\
\hspace*{2mm} 2. \textbf{Return} $T'$; \\
\end{quote}
\vspace*{-10mm}
\strut}
$$\boxit{\box4}$$
\vspace*{-9mm}
\caption{The deterministic algorithm for \textsc{Parameterized Flip Distance}} \label{fig3}
\end{figure}

The following theorem proves the correctness of the algorithm $\textbf{FLIPDT}$.

\begin{theorem}
Let $(T_{start}, T_{end}, k)$ be an input instance.  \textbf{FLIPDT} is correct and runs in time $O^{*}(k\cdot 32^{k})$.
\end{theorem}
\begin{proof}
Suppose that $(T_{start}, T_{end}, k)$ is a yes-instance.  There must exist a sequence of flips $F$ of length $k$ such that $T_{start}\xrightarrow{F} T_{end}$.  Thus $D_{F}$ exists according to Definition~\ref{def1}. By \textbf{NDTRV} and Lemma~\ref{forest}, there exists an undirected graph $G$ consisting of a set of vertex-disjoint track trees $A_{1},...,A_{t}$. Moreover, Theorem~\ref{main} shows that there exists a sequence of actions $Q$ following which we can perform all flips of $D_{F}$ in a topologically-sorted order. Due to \textbf{NDTRV}, $Q$ consists of several subsequences $seq_{1},...,seq_{t}$, in which $seq_{i}$ is constructed in the $i$-th iteration and corresponds to the track tree $A_{i}$ for $i=1,...,t$.  Supposing the size of $A_{i}$ is $\lambda_{i}$ for $i=1,...,t$ satisfying $\lambda_{1}+...+\lambda_{t}=k$, $seq_{i}$ contains $\lambda_{i}$ actions of type (ii) corresponding to the nodes of $A_{i}$ as well as $\lambda_{i}-1$ actions of type (i) corresponding to the edges of $A_{i}$. \textbf{FLIPDT} guesses the size of every track tree by enumerating all possible partitions of $k$ into $(k_{1},...,k_{t})$ such that $k_{1}+...+k_{t}=k$ and $k_{1},...,k_{t}\geq 1$.  We say that $k_{i}$ is distributed to the $i$-th iteration or the distribution for the $i$-th iteration is $k_{i}$ for $i=1,...,t$.

We claim that \textbf{FLIPDT} is able to perform a sequence $\Sigma$ of actions which correctly guesses every subsequence $seq_{1},...,seq_{t}$ of the objective sequence $Q$, that is, $\Sigma$ is a concatenation of $seq_{1},...,seq_{t}$.  Suppose that \textbf{FLIPDT} has completed $i$ iterations.  We prove this claim by induction on $i$.  At the first iteration, \textbf{FLIPDT} starts by picking the first necessary edge $e_{1}$ in list $O_{lex}$.  In the first iteration of constructing $Q$, \textbf{NDTRV} starts by picking an arbitrary necessary edge. Without loss of generality, it chooses $e_{1}$ and construct $seq_{1}$ starting from $e_{1}$.  The length of $seq_{1}$ is $2\lambda_{1}-1$.   Since \textbf{FLIPDT} tries every distribution in $\{1,...,k\}$ for the first iteration and $1\leq \lambda_{1}\leq k$, there is a correct guess of the distribution equal to $\lambda_{1}$ for this iteration.  Under this correct guess, \textbf{FLIPDT} tries all possible sequences of actions of length $2\lambda_{1}-1$ starting from $e_{1}$. It follows that \textbf{FLIPDT} is able to perform a sequence that is equals to $seq_{1}$ in the first iteration resulting in a triangulation $T_{1}$.

Suppose that the claim is true for any first $i$ iterations ($1\leq i<t$). That is, under some guess for the partition of $k$, $\lambda_{1},...,\lambda_{i}$ are distributed to the first $i$ iterations respectively.  Moreover, \textbf{FLIPDT} has completed $i$ iterations and performed a sequence of actions $seq_{concat,i}$, which is equal to the concatenation of $seq_{1},...,seq_{i}$, resulting in a triangulation $T_{i}$. Based on $T_{i}$ and $seq_{concat,i}$, \textbf{FLIPDT} is ready to perform the $(i+1)$-th iteration.  Suppose that \textbf{FLIPDT} picks $e_{i+1}$ from $O_{lex}$.  Let us see the construction of $Q$ in \textbf{NDTRV}. Suppose \textbf{NDTRV} has constructed the first $i$ track trees $A_{1},...,A_{i}$, and it is ready to begin a new iteration by arbitrarily picking a necessary edge in the current triangulation.  Since \textbf{FLIPDT} correctly guessed and performed the first $i$ subsequences of $Q$, $T_{i}$ is exactly equal to the current triangulation in \textbf{NDTRV}.  Thus $e_{i+1}$ is a candidate edge belonging to the set of all selectable necessary edges for \textbf{NDTRV} in this iteration.  Without loss of generality, it chooses $e_{i+1}$ and constructs $seq_{i+1}$ of length $2\lambda_{i+1}-1$ starting from $e_{i+1}$.  Since the sizes of $A_{1},...,A_{i}$ are $\lambda_{1},...,\lambda_{i}$ respectively, we get that $1\leq \lambda_{i+1}\leq k-(\lambda_{1}+...+\lambda_{i})$.  We argue that \textbf{FLIPDT} is able to perform a sequence that is equal to the concatenation of $seq_{1},...,seq_{i+1}$. Since the edges in $O_{lex}$ are ordered lexicographically and \textbf{FLIPDT} chooses necessary edges in a fixed manner, \textbf{FLIPDT} is sure to choose $e_{i+1}$ to begin the $(i+1)$-th iteration for every guessed sequence in which the first $i$ subsequences are equal to $seq_{1}$,...,$seq_{i}$ respectively.  Thus \textbf{FLIPDT} actually tries every distribution in $\{1,...,k-(\lambda_{1}+...+\lambda_{i}\})$ for the $(i+1)$-th iteration starting from $e_{i+1}$ based on $T_{i}$ and $seq_{concat,i}$.  It follows that there is a correct guess of distribution for the $(i+1)$-th iteration which is equal to $\lambda_{i+1}$.  Under this correct guess of distribution, \textbf{FLIPDT} tries all possible sequences of length $2\lambda_{i+1}-1$ starting from $e_{i+1}$ on $T_{i}$ based on $seq_{concat,i}$, ensuring that one of them is equal to $seq_{i+1}$.  It follows that the claim is true for the first $i+1$ iterations. This completes the inductive proof for the claim.

If $(T_{start},T_{end},k)$ is a yes-instance, the action sequence $Q$ of length at most $2k$ exists and the deterministic algorithm can find such a sequence.  Otherwise, there is no valid sequence $F$ of length $k$.  Thus there is no such action sequence $Q$.  As a result, \textbf{FLIPDT} returns NO.  It is proved that \textbf{FLIPDT} decides the given instance $(T_{start},T_{end},k)$ correctly.

Finding and ordering all necessary edges in $T_{start}$ takes $O(n+k\log k)$ time, and \textbf{FLIPDT} may update the ordering $O_{lex}$ at the beginning of each iteration. The number of partitions of $k$ is known as the composition number of $k$, which is $2^{k-1}$.  Under each partition $(k_{1},...,k_{t})$ of $k$ and for each $k_{i}$, $i=1,...,t$, we enumerate all possible subsequences of actions in which there are $k_{i}$ actions of type (ii) and $k_{i}-1$ actions of type (i).  It follows that the number of all possible subsequences is bounded by  $\binom{2(k_{i}-1)}{k_{i}-1}\times 4^{k_{i}-1}=O^{*}(16^{k_{i}})$ since there are four choices for action (i) and one choice for action (ii).  Here we use Stirling's approximation $n!\approx \sqrt{2\pi n}(n/e)^{n}$ and get that $\binom{2(k_{i}-1)}{k_{i}-1}=O^{*}(4^{k_{i}})$.  It follows that there are $O^{*}(16^{k_{1}})\times O^{*}(16^{k_{2}})\times...\times O^{*}(16^{k_{t}})=O^{*}(16^{k})$ cases under each partition.  Since for each case we can perform the sequence of actions in $O(k)$ time, and the resulting triangulation can be compared to $T_{end}$ in $O(k)$ time, the running time of the whole algorithm is bounded by $O^{*}(k\cdot 2^{k-1}\cdot (n+k\log k)+k\cdot 2^{k-1}\cdot 16^{k})=O^{*}(k\cdot 32^{k})$.

According to the definition of \textsc{Parameterized Flip Distance}, we need to check if we can find a shorter valid sequence for the given triangulations $T_{start}$ and $T_{end}$.  This is achieved by calling \textbf{FLIPDT} on each instance $(T_{start}, T_{end}, k')$ for $k'=0,...,k$.  The running time is bounded by $\sum_{k'=0}^k O^{*}(k'\cdot 32^{k'})=O^{*}(k\cdot 32^{k})$.
\end{proof}
\section{Conclusion}
In this paper we presented an FPT algorithm running in time $O^{*}(k\cdot 32^{k})$ for \textsc{Parameterized Flip Distance}, improving the previous $O^{*}(k\cdot c^{k})$-time ($c\leq 2\times 14^{11}$) FPT algorithm by Kanj and Xia \cite{Kanj}. An important related problem is computing the flip distance between triangulations of a convex polygon, whose traditional complexity is still unknown. Although our algorithm can be applied to the case of convex polygon, it seems that an $O(c^{k})$ algorithm with smaller $c$ for this case probably exists due to its more restrictive geometric property. In addition, whether there exists a polynomial kernel for \textsc{Parameterized Flip Distance} is also an interesting problem.

\appendix
\bibliography{flipdistance}

\end{document}